\documentclass[pdflatex,sn-mathphys-num]{sn-jnl}

\usepackage{graphicx}%
\usepackage{multirow}%
\usepackage{amsmath,amssymb,amsfonts}%
\usepackage{amsthm}%
\usepackage{mathrsfs}%
\usepackage[title]{appendix}%
\usepackage{xcolor}%
\usepackage{textcomp}%
\usepackage{manyfoot}%
\usepackage{float}       
\usepackage{booktabs} 
\usepackage{algorithm}%
\usepackage{algorithmicx}%
\usepackage{algpseudocode}%
\usepackage{listings}%

\theoremstyle{thmstyleone}%
\newtheorem{theorem}{Theorem}
\theoremstyle{thmstyletwo}%

\theoremstyle{thmstylethree}%

\newtheorem{law}{Law}

\begin{document}

\title[Composition Law for Conjugate Observables]{Composition Law of Conjugate Observables in Random Permutation Sorting Systems}
\subtitle{A Foundation for Ubiquitous True Uniform Randomness}


\author{Yurang R. Kuang (Preferred name: Randy Kuang; ORCID: \href{https://orcid.org/0000-0002-5567-2192}{0000-0002-5567-2192})}

\affil*[1]{\orgdiv{Research}, \orgname{Quantropi Inc.}, \orgaddress{\street{1545 Carling Av., Suite 620}, \city{Ottawa}, \postcode{K1Z 8P9}, \state{ON}, \country{Canada}}}


\abstract{
We report the discovery of a fundamental \emph{composition law} governing conjugate observables in the Random Permutation Sorting System (RPSS), which quantifies how discrete permutation counts shape continuous timing distributions. The system is characterized by two conjugate observables: the discrete permutation count $\hat{N}_p$ and the continuous elapsed time $\hat{T}$ per sorting cycle. Their relationship is governed by the composition law
\[
\varphi_{\hat{T}}(\omega) = G_{\hat{N}_p}(\varphi_X(\omega)),
\]
linking the characteristic function of elapsed time to the probability generating function of permutation counts. This law enables \emph{entropy purification}, whereby ubiquitous microarchitectural timing variations are transformed into uniform randomness through geometric convergence. We establish synchronous convergence theorems with explicit bounds, providing provable guarantees of cryptographic uniformity. Empirical validation across diverse computing platforms confirms Shannon entropy consistently exceeding 7.9998 bits per byte, chi-square uniformity within statistical confidence bounds, and robustness to environmental perturbations. By instantiating this law, RPSS enables general-purpose computing devices to serve as self-contained sources of provably uniform randomness, independent of specialized hardware. The Composition Law establishes computational conjugate observables as a universal foundation for trustworthy randomness generation, securing cryptographic purity from emergent computational chaos.
}




\maketitle

\section{Introduction}

Randomness lies at the foundation of modern science and technology, underpinning cryptographic security, post-quantum cryptography (PQC), quantum simulation, statistical sampling, and Monte Carlo methods. In cryptography, the strength of protocols is inseparably tied to the quality of their randomness: weak or biased random sources have led to catastrophic key recovery attacks and protocol breaks. This reliance is even more acute in PQC. Leading lattice-based standards such as \emph{Kyber} (KEM)~\cite{KYBER} and \emph{Dilithium} (digital signatures)~\cite{DILITHIUM}, recently selected by NIST for standardization~\cite{NIST-FIPS-203, NIST-FIPS-204}, consume massive amounts of fresh randomness during key generation, encryption, and signature operations. Their security requires not only unpredictability but also statistical uniformity across billions of random draws, thereby amplifying the demand for robust, universal entropy sources. 

The ideal of \emph{true randomness} requires both unpredictability and uniform distribution, properties that naturally arise in physical phenomena. Conventional True Random Number Generators (TRNGs) rely on specialized hardware—drawing from electronic noise~\cite{trng-yang-2018, trng-tehranipoor-2023, trng-ansari-2022, trng-ji-2020}, chaotic dynamics~\cite{rng-chaotic-2024, rng-chaotic-lu-2025, rng-chaotic-2019, rng-chaotic-Bonilla2016}, or quantum effects~\cite{qrng-Ma2016, qrng-Gabriel2010, zhang2021qrng}. While powerful, these approaches are bound to physical implementations and frequently require bias correction to achieve uniformity~\cite{nist80090a,nist80090b}, introducing complexity and potential vulnerabilities. As a result, hardware-based solutions have struggled to provide universal, verifiable, and trustworthy randomness across diverse computational environments~\cite{Rahman2014TI-TRNG}.

Modern computing platforms, despite being deterministic in principle, exhibit rich stochastic dynamics at the microarchitectural level. Cache hierarchies, branch prediction, memory contention, and operating system scheduling introduce timing variability typically treated as performance noise. Software-based entropy harvesters exploit such behaviors~\cite{linux-Gutterman2006, IntelRNG2012, maurer2019jitterentropy, fischer2002trng}, but usually as seed material for pseudorandom generators rather than as direct sources of true randomness. This raises a central question: \emph{can intrinsic computational fluctuations themselves be transformed into a mathematically provable foundation for true random number generation, eliminating dependence on specialized hardware and statistical post-processing?}

Here we introduce the \emph{Random Permutation Sorting System (RPSS)}, a computational framework that establishes a new paradigm for randomness generation. In RPSS, a disordered integer array is subjected to repeated random permutations—conceptually analogous to a Quantum Permutation Pad (QPP)~\cite{qpp-springer-kuang-2022}—until convergence. The core operation is defined by:
\begin{equation}\label{eq:qpp}
\hat{p}^{-1} = \prod_{j=1}^{n_p} \hat{p}_j,
\end{equation}
where $\hat{p}_j$ are random permutations and $n_p$ is the permutation count. This QPP-based approach, previously validated as QPP-RNG~\cite{kuang-qpp-rng-icccas} against NIST SP 800-90B~\cite{qpp-rng-sci-kuang-2025}, provides the experimental foundation for our theoretical generalization to RPSS.

RPSS reveals a pair of \textbf{conjugate observables}~\cite{kuang2025-rpss-arxiv}: the discrete permutation count $\hat{N}_p$ and the continuous elapsed permutation runtime $\hat{T}$ per sorting cycle. The core of this framework is our discovery of a fundamental \emph{composition law},
\begin{equation}\label{eq:law}
\varphi_{\hat{T}}(\omega) = G_{\hat{N}_p}(\varphi_X(\omega)),
\end{equation}
which links the probability-generating function (PGF) of $\hat{N}_p$ to the characteristic function of $\hat{T}$. This law fully characterizes the compound stochastic structure of RPSS, unifying discrete combinatorial randomness and continuous timing variability within a single theoretical framework.

From this composition law, we establish explicit proofs of \emph{synchronous convergence to uniformity} under modular reduction, with geometric bounds that provide provable guarantees of randomness. This result elevates timing fluctuations—long viewed as nuisance variability—into a constructive source of cryptographic-grade entropy. Experiments confirm that RPSS achieves Shannon entropy consistently exceeding $7.9998$ bits per byte and robust uniformity across diverse environments, including mobile and embedded platforms, without reliance on specialized hardware or external randomness extractors.

Our findings position computational conjugate observables as a universal foundation for trustworthy randomness generation. By demonstrating that the complexity of modern processors intrinsically encodes the mathematical structure required for true randomness, we provide a pathway for post-quantum cryptography, IoT, and secure communications to access cryptographic-grade entropy without additional hardware. Importantly, RPSS can be naturally embedded into post-quantum cryptographic modules, closing the entropy gap in schemes such as lattice-based KEMs and signatures where massive volumes of high-quality randomness are consumed. Ultimately, RPSS reveals that true randomness is not confined to physical noise sources, but is a fundamental emergent property of computation itself for entropy purification.

\section{Results}

\subsection{Definition and Compound Structure}
\label{subsec:definition}

The Random Permutation Sorting System (RPSS)~\cite{kuang2025-rpss-arxiv} is characterized by a pair of conjugate observables: the discrete permutation count $\hat{N}_p$ and the continuous total permutation runtime $\hat{T}$.  

The permutation count $\hat{N}_p$ follows a negative binomial distribution, denoted $\hat{N}_p \sim \mathrm{NB}(m, p)$, representing the number of attempts required to achieve $m$ successful sorts, where each attempt succeeds with probability $p = 1/N!$. While earlier work~\cite{kuang2025-rpss-arxiv} analyzed $\hat{N}_p$ in isolation, here it functions as the stochastic index for elapsed-time analysis.

The total elapsed time until the $m$-th successful sort is defined as the compound random variable
\begin{equation}\label{eq:T}
\hat{T} = \sum_{j=1}^{\hat{N}_p} X_j,
\end{equation}
where $X_j$ are i.i.d.\ positive random variables representing the runtime of each permutation attempt. These capture microarchitectural and system-level effects, including cache misses, branch mispredictions, and OS scheduling delays. We assume $X_j$ has finite mean $\mu_X$ and variance $\sigma_X^2$, with $\sigma_X$ often comparable to or exceeding $\mu_X$ in jitter-dominated environments.

This formulation establishes $\hat{T}$ as a canonical \emph{compound distribution}, integrating algorithmic randomness ($\hat{N}_p$) with physical variability ($X_j$). It provides the theoretical foundation for proving convergence of both observables under modular reduction, representing the core stochastic mechanism of RPSS.

\subsection{Composition Law of Conjugate Observables}
\label{sec:composition-law}

The RPSS admits two natural observables:  
(i) the discrete permutation count $\hat{N}_p$, and  
(ii) the continuous elapsed time $\hat{T}$.  
These observables form a conjugate pair, analogous to momentum and position in physical systems.  

\begin{law}[Composition Law of Conjugate Observables]
\label{law:composition}
In the RPSS, the elapsed time $\hat{T}$ and permutation count $\hat{N}_p$ are related by:
\begin{equation}
\varphi_{\hat{T}}(\omega) = G_{\hat{N}_p}(\varphi_X(\omega)),
\end{equation}
where $\varphi_{\hat{T}}(\omega)$ is the characteristic function (CF) of $\hat{T}$, 
$G_{\hat{N}_p}(z)$ is the probability generating function (PGF) of $\hat{N}_p$, 
and $\varphi_X(\omega)$ is the CF of per-permutation runtime $X$.
\end{law}

For $\hat{N}_p \sim \mathrm{NB}(m,p)$ with $p=1/N!$, this yields the explicit form:
\begin{equation}
\varphi_{\hat{T}}(\omega) = \left( \frac{p \, \varphi_X(\omega)}{1 - (1-p)\,\varphi_X(\omega)} \right)^m,
\end{equation}
demonstrating that the distribution of elapsed permutation time is obtained by mapping individual permutation runtimes into the discrete permutation-counting structure.

\paragraph*{Interpretation.}  
Law~\ref{law:composition} establishes an intrinsic \emph{composition principle}: the stochastic structure of $\hat{N}_p$ is systematically mapped into that of $\hat{T}$ via runtime fluctuations. This principle governs RPSS dynamics and explains how microarchitectural noise transforms into true randomness.

\paragraph*{Significance.}  
By framing RPSS randomness under a law-like statement, we define a device-independent mechanism. Much like fundamental physical laws formalize natural phenomena, the composition law formalizes the transformation between conjugate observables in computational systems, providing the foundation for provable convergence to cryptographic uniformity.

\subsubsection{Derivation via Characteristic Functions}
\label{subsubsec:cf-correctness}

We prove the composition law through the characteristic function formalism:
\begin{equation}
\varphi_{\hat{T}}(\omega) = \mathbb{E}[e^{i \omega \hat{T}}] = \mathbb{E}\!\left[ \mathbb{E}[e^{i \omega \hat{T}} \mid \hat{N}_p ] \right].
\end{equation}
For $\hat{N}_p = k$, we have $\hat{T} = \sum_{j=1}^k X_j$, giving
\begin{equation}
\mathbb{E}[ e^{i \omega \hat{T}} \mid \hat{N}_p = k ] = \big( \varphi_X(\omega) \big)^k.
\end{equation}
Substituting yields the fundamental relation:
\begin{equation}
\varphi_{\hat{T}}(\omega) = \mathbb{E}[ (\varphi_X(\omega))^{\hat{N}_p} ] = G_{\hat{N}_p}(\varphi_X(\omega)),
\end{equation}
which recovers the composition law. For $\hat{N}_p \sim \mathrm{NB}(m,p)$, the PGF is known to be
\begin{equation}
G_{\hat{N}_p}(z) = \left( \frac{pz}{1-(1-p)z} \right)^m, \quad |z|<1/(1-p),
\end{equation}
providing the explicit characteristic function for $\hat{T}$.

\subsubsection{Moment Structure of the Compound Distribution}
\label{subsubsec:moments}

The compound structure of $\hat{T} = \sum_{j=1}^{\hat{N}_p} X_j$ produces rich statistical properties that combine combinatorial and microarchitectural randomness.

\paragraph{Moments of $\hat{N}_p$.}
For $\hat{N}_p \sim \mathrm{NB}(m, p)$, the permutation count exhibits inherent overdispersion and higher-order structure:
\begin{align}
\mathbb{E}[\hat{N}_p] &= \tfrac{m}{p}, &
\mathrm{Var}[\hat{N}_p] &= \tfrac{m (1-p)}{p^2}, \\
\kappa_3(\hat{N}_p) &= \tfrac{m (1-p)(1-2p)}{p^3}, &
\kappa_4(\hat{N}_p) &= \tfrac{m (1-p)(1 - 6p(1-p))}{p^4}.
\end{align}
These moments capture the skewness and overdispersion inherent in the permutation counting process, providing the deterministic foundation for empirical validation.

\paragraph{Mean and Variance of $\hat{T}$.}
Applying the law of total expectation and variance:
\begin{align}
\mathbb{E}[\hat{T}] &= \frac{m \mu_X}{p}, \\
\mathrm{Var}[\hat{T}] &= \frac{m \sigma_X^2}{p} + \frac{m(1-p)\mu_X^2}{p^2}.
\end{align}
This decomposition clearly separates contributions from microscopic runtime fluctuations ($\sigma_X^2$) and the stochastic number of permutations ($\mathrm{Var}[\hat{N}_p]$).

\paragraph{Higher Cumulants of $\hat{T}$.}
By the law of total cumulance~\cite{brillinger-1969}, the complete moment structure is:
\begin{align}
\kappa_1(\hat{T}) &= \mu_X \, \kappa_1(\hat{N}_p), \\
\kappa_2(\hat{T}) &= \mu_X^2 \kappa_2(\hat{N}_p) + \sigma_X^2 \kappa_1(\hat{N}_p), \\
\kappa_3(\hat{T}) &= \mu_X^3 \kappa_3(\hat{N}_p) + 3 \mu_X \sigma_X^2 \kappa_2(\hat{N}_p) + \kappa_3(X)\kappa_1(\hat{N}_p), \\
\kappa_4(\hat{T}) &= \mu_X^4 \kappa_4(\hat{N}_p) + 6 \mu_X^2 \sigma_X^2 \kappa_3(\hat{N}_p) \nonumber \\
&\quad + (4 \mu_X \kappa_3(X) + 3 \sigma_X^4)\kappa_2(\hat{N}_p) + \kappa_4(X)\kappa_1(\hat{N}_p).
\end{align}
These cumulants allow a full description of the shape of $\hat{T}$, including tail behavior and asymmetry.

\paragraph{Statistical Implications.}
The resulting skewness and kurtosis of $\hat{T}$,
\begin{equation}
\gamma_1(\hat{T}) = \frac{\kappa_3(\hat{T})}{\kappa_2(\hat{T})^{3/2}}, \quad
\gamma_2(\hat{T}) = \frac{\kappa_4(\hat{T})}{\kappa_2(\hat{T})^2},
\end{equation}
are distinct from those of $\hat{N}_p$ due to the compounding with $X_j$. For $\hat{N}_p$, the corresponding measures are:
\begin{equation}
\gamma_1(\hat{N}_p) = \frac{2-p}{\sqrt{m(1-p)}}, \quad
\gamma_2(\hat{N}_p) = \frac{6}{m} + \frac{p^2}{m(1-p)}.
\end{equation}
This difference explains the observed "fat-to-skinny" runtime distributions of $\hat{T}$ compared to the more predictable shape of $\hat{N}_p$. The overdispersion and higher moments of $\hat{N}_p$ translate into the enhanced fat tails and complex shape that make $\hat{T}$ suitable for high-entropy cryptographic applications. Crucially, the deterministic properties of $\hat{N}_p$ provide the theoretical anchor for empirical validation of $\hat{T}$, even without explicit knowledge of $\varphi_X(\omega)$.

\subsection{Synchronous Convergence to Uniformity}
\label{subsec:convergence}

\begin{theorem}[Synchronous Convergence of RPSS Observables]
\label{thm:synchronous-convergence}
Let $\hat{N}_p \sim \mathrm{NB}(m, p)$ with $p = 1/N!$, and let 
\[
\hat{T} = \sum_{j=1}^{\hat{N}_p} X_j,
\]
where $X_j$ are i.i.d.\ positive random variables with finite mean $\mu_X$ and variance $\sigma_X^2$. Under reduction modulo $R$, as $M = m N! \to \infty$, both observables converge to the discrete uniform distribution with exponential convergence rates:
\begin{align}
\Pr\big[\hat{N}_p \bmod R = k\big] &= \frac{1}{R} + \mathcal{O}(\rho_N^m), \\
\Pr\big[\hat{T} \bmod R = k\big] &= \frac{1}{R} + \mathcal{O}(\rho_T^m),
\end{align}
for $k = 0,1,\dots,R-1$, where $0 < \rho_N, \rho_T < 1$ are geometric decay constants.
\end{theorem}

\begin{proof}
By the composition law (Law~\ref{law:composition}), the characteristic function of $\hat{T}$ is
\[
\varphi_{\hat{T}}(\omega) = \left( \frac{\varphi_X(\omega)}{(1 - \varphi_X(\omega)) N! + \varphi_X(\omega)} \right)^m.
\]
Discrete Fourier inversion for the modulo-$R$ distribution gives
\[
\Pr[\hat{T} \bmod R = r] = \frac{1}{R} + \frac{1}{R} \sum_{k=1}^{R-1} \varphi_{\hat{T}}(\omega_k) e^{-i r \omega_k}, 
\quad \omega_k = \frac{2\pi k}{R}.
\]
Defining the geometric factors:
\begin{align}
\rho_{N,k} &= \left| \frac{1}{(1 - e^{i\omega_k}) N! + e^{i\omega_k}} \right|, \\
\rho_{T,k} &= \left| \frac{\varphi_X(\omega_k)}{(1 - \varphi_X(\omega_k)) N! + \varphi_X(\omega_k)} \right|,
\end{align}
and setting $\rho_N = \max_k \rho_{N,k}$, $\rho_T = \max_k \rho_{T,k}$, we obtain the error bounds:
\begin{align}
\left| \Pr[\hat{N}_p \bmod R = r] - \frac{1}{R} \right| &\le \frac{R-1}{R} \rho_N^m, \\
\left| \Pr[\hat{T} \bmod R = r] - \frac{1}{R} \right| &\le \frac{R-1}{R} \rho_T^m.
\end{align}
Since $M = m N! \to \infty$ forces exponential decay of both bounds, synchronous convergence is established, governed by the slower rate $\max(\rho_N^m, \rho_T^m)$.
\end{proof}

\paragraph{Key Implication.}
The exponential decay $\mathcal{O}(\rho^m)$ with $\rho = \max(\rho_N, \rho_T)$ quantifies the RPSS uniformity guarantee. While $\rho_N$ is computable analytically, $\rho_T$ depends on the unknown distribution of $X_j$. Empirical validation confirms that parameter sets satisfying $\rho_N^m < 0.01$ ensure synchronous convergence in practice.

Table~\ref{tab:optimal-parameters} summarizes empirically validated parameters that ensure $\rho_N^m < 0.01$ while maintaining computational efficiency, with cost per byte given by $C_{\mathrm{byte}} = \frac{8}{n} \cdot m N! N$.

\begin{table}[htb]
\centering
\caption{Empirically validated parameters for $n$-bit uniform residues ensuring $\rho_N^m < 0.01$.}
\label{tab:optimal-parameters}
\small
\begin{tabular}{r r r r r r r}
\toprule
$n$ & $R=2^n$ & $N$ & $m$ & $N!$ & $M=mN!$ & $\rho_N^m$ \\
\midrule
1 & 2 & 2 & 12 & 2 & 24 & 0.0003 \\
1 & 2 & 3 & 3 & 6 & 18 & 0.010 \\
2 & 4 & 3 & 5 & 6 & 30 & 0.004 \\
2 & 4 & 4 & 2 & 24 & 48 & 0.010 \\
4 & 16 & 4 & 4 & 24 & 96 & 0.001 \\
4 & 16 & 5 & 2 & 120 & 240 & 0.006 \\
8 & 256 & 5 & 5 & 120 & 600 & 0.001 \\
\bottomrule
\end{tabular}
\end{table}

\subsection{Empirical Runtime Distribution of Individual Permutations}
\label{subsec:runtime-distribution}

The per-permutation runtime $X_j$ in the Random Permutation Sorting System (RPSS) exhibits intrinsic stochastic variability, even under controlled system conditions without external perturbations. Empirical measurements for array sizes $N = 4$--$7$, summarized in Table~\ref{tab:empirical-runtimes-corrected}, reveal several consistent patterns that challenge conventional statistical models:

\begin{itemize}
    \item \textbf{Sharp concentration at minimal runtimes}: Most permutations complete within 0--2 ticks for $N \leq 7$, with distributions sharply peaked at minimal values, indicating highly optimized execution paths under favorable conditions.
    
    \item \textbf{Heavy-tailed behavior with extreme outliers}: Rare permutations experience runtime spikes (``Too Big'' events, $\gtrsim 500$ ticks) orders of magnitude larger than modal values. These extreme events persist under isolated conditions and originate from intrinsic microarchitectural phenomena including cache hierarchy misses, branch mispredictions, and pipeline stalls.
    
    \item \textbf{Discrete and multimodal structure}: Runtime distributions display pronounced spikiness with uneven probability mass across tick values. Intermediate runtimes (3--10 ticks) are often sparsely populated, creating a separation between the dominant low-runtime mode and extreme tail events.
\end{itemize}

\begin{table}[htb]
\centering
\caption{Empirical per-permutation runtimes $X_j$ (ticks) for array sizes $N=4$--$7$. Statistics include sample size (×1000), mean $\mu_X$, standard deviation $\sigma_X$, and frequency counts by tick value. Sample sizes below 1000 indicate ``Too Big'' events exceeding measurement thresholds.}
\begin{tabular}{c| c| c| c| l}
\hline
$N$ & Samples & $\mu_X$ & $\sigma_X$ & Frequencies (ticks: 0,1,2,...) \\
\hline
7 & 998 & 1.566 & 0.585 & \texttt{0,453,538,2,2,0,1,2} \\
7 & 1000 & 1.529 & 0.584 & \texttt{0,490,502,4,2,0,0,1,1} \\
7 & 994 & 2.165 & 1.05 & \texttt{0,156,583,241,6,1,1,1,1,0,1,0,1}\\
  &     &       &      & \texttt{,0,0,0,0,1,1} \\
7 & 996 & 3.147 & 0.676 & \texttt{0,0,11,873,91,12,2,2,3,1,0,1} \\
6 & 999 & 1.359 & 0.618 & \texttt{0,675,311,3,4,1,4,1} \\
6 & 994 & 1.485 & 0.771 & \texttt{0,564,410,11,4,1,1,0,1,0,0,1,0,1} \\
6 & 993 & 1.871 & 0.876 & \texttt{0,307,607,23,17,36,2,1} \\
6 & 999 & 1.472 & 0.619 & \texttt{0,559,426,8,4,0,0,1,0,0,0,0,0,1} \\
6 & 991 & 9.98 & 13.6 & \texttt{0,0,0,0,0,0,0,0,504,427,8,21,17,} \\
 &     &       &      & \texttt{4,3,4,2,1,0,1} \\
5 & 997 & 1.154 & 0.398 & \texttt{0,853,137,5,1,1} \\
5 & 993 & 1.632 & 0.833 & \texttt{0,506,385,86,6,4,2,3,0,0,1} \\
5 & 996 & 1.174 & 0.496 & \texttt{0,491,479,15,4,1,4,2} \\
5 & 997 & 1.148 & 0.387 & \texttt{0,854,136,2,4,0,1} \\
5 & 999 & 1.489 & 0.598 & \texttt{0,249,718,25,3,3,0,1} \\
4 & 999 & 0.973 & 0.258 & \texttt{42,946,8,2,1} \\
4 & 998 & 1.022 & 0.522 & \texttt{34,935,18,6,3,0,1,0,0,0,0,0,0,1} \\
4 & 998 & 1.346 & 0.781 & \texttt{0,747,187,56,2,2,3,0,1} \\
4 & 998 & 1.345 & 0.772 & \texttt{18,772,191,8,4,1,1,2,1} \\
\hline
\end{tabular}
\label{tab:empirical-runtimes-corrected}
\end{table}

\paragraph{Cryptographic Entropy Strength and Attack Resistance}
The statistical variability in $X_j$ constitutes a cryptographically robust entropy source resistant to modeling attacks. Heavy-tailed behavior and extreme outliers ($\gtrsim 500$ ticks) generate distributions that resist parametric characterization, complicating adversary efforts to construct accurate timing models—a property consistent with established cryptographic principles for robust entropy sources. This complexity is amplified by run-to-run variability in both mean ($\mu_X$ varying by 7.35× for $N=6$) and variance ($\sigma_X$ varying by 22×), ensuring computational infeasibility of predicting individual runtime instances even with algorithmic knowledge. This property is essential for resisting statistical attacks reliant on distributional predictability.

\paragraph{Statistical Variability as Entropy Source}
The fundamental unpredictability of RPSS emerges from substantial run-to-run variability in mean runtime $\mu_X$ and variance $\sigma_X^2$, even for identical array sizes $N$, representing the core entropy source:

\begin{itemize}
    \item \textbf{Mean runtime variability}: For $N=7$, $\mu_X$ ranges from 1.529 to 3.147 ticks (2.06× variation). For $N=6$, the range extends from 1.359 to 9.98 ticks (7.35× variation), with one run exhibiting order-of-magnitude increase due to systematic slowdown.
    
    \item \textbf{Variance instability}: Standard deviation $\sigma_X$ shows greater relative variability. For $N=6$, $\sigma_X$ ranges from 0.618 to 13.6 (22× variation), indicating fundamentally different runtime distributions across experimental runs.
    
    \item \textbf{Unpredictable distribution shapes}: Frequency patterns reveal completely different modal structures across runs—some concentrated at 1-2 ticks, others showing broader distributions or secondary peaks.
\end{itemize}

This statistical variability reflects inherent non-determinism in modern microprocessor execution, where identical computational tasks exhibit dramatically different timing characteristics due to microarchitectural state variations, rather than measurement error or external factors.

\paragraph{Extreme Event Characterization}
Extreme runtime events ($\gtrsim 500$ ticks) under nominal operating conditions demonstrate the fundamentally heavy-tailed nature of $X_j$. Sample size discrepancies (values below 1000 in Table~\ref{tab:empirical-runtimes-corrected}) reflect ``Too Big'' events excluded from measurement. This heavy-tailed property directly shapes the compound elapsed-time distribution $\hat{T}$:

\begin{enumerate}
    \item \textbf{Tail inheritance and amplification}: Heavy tails of $X_j$ propagate through the random sum $\hat{T} = \sum_{j=1}^{\hat{N}_p} X_j$, producing leptokurtic behavior observed in Section~\ref{subsubsec:t-dist}. Rare extreme runtimes can dominate total elapsed time.
    
    \item \textbf{Partial smoothing through aggregation}: Despite heavy tails, summation over $\hat{N}_p$ permutations provides statistical smoothing. The central limit theorem for random sums ensures $\hat{T}$ exhibits more stable central behavior while preserving essential stochastic variability.
    
    \item \textbf{Modeling implications}: The empirical distribution of $X_j$ resists simple parametric characterization, resembling a discrete mixture with high probability mass at minimal runtimes and subexponential tails capturing rare extreme events. Standard distributions fail to simultaneously capture sharp concentration and heavy-tailed behavior.
\end{enumerate}

\paragraph{Implications for Entropy Extraction}
The intrinsic statistical variability of $X_j$—particularly unpredictable fluctuations in $\mu_X$ and $\sigma_X^2$ across runs under identical RPSS configurations—provides robust entropy extraction. This ensures substantial unpredictable timing behavior even under identical external conditions and computational tasks.

Crucially, as established in Section~\ref{subsubsec:cf-correctness}, this runtime variability—while shaping $\hat{T}$ morphology—does not compromise final output uniformity guarantees. The conjugacy relationship ensures uncertainty in the continuous time domain translates to reliable uniformity in the discrete count domain after modular reduction.

In summary, RPSS per-permutation runtime $X_j$ exhibits intrinsic statistical variability in central tendency and dispersion that underpins system entropy generation, while compound structure and modular reduction ensure reliable uniform TURNG output.

\subsection{Raw Distribution Features and Moment Validation}
\label{subsec:raw-distributions}

\subsubsection{Permutation Count Distribution}
\label{subsubsec:np-dist}

\begin{figure}[H]
\centering
\includegraphics[scale=0.4]{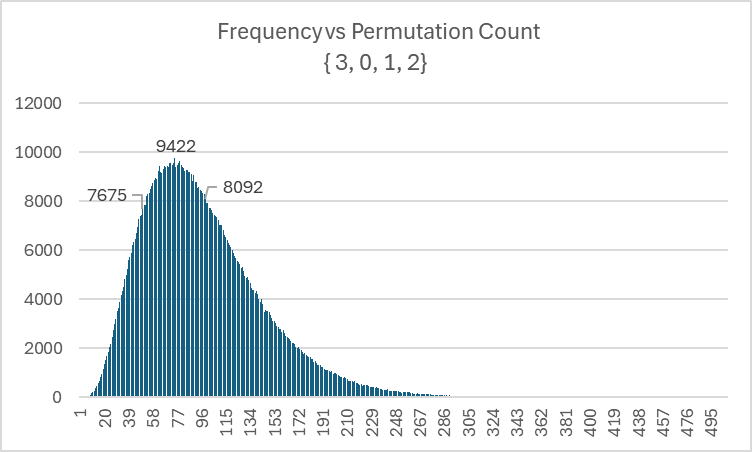}
\caption{Empirical distribution of permutation count $\hat{N}_p$ for ($N=4, m=4$). }
\label{fig:np-dist}
\end{figure}

Our theoretical model posits that permutation sorting trial count $\hat{N}_p$ follows a \textbf{Negative Binomial distribution}. Validation experiments with parameters $N=4$ and $m=4$ (theoretical success probability $p = 1/N! = 1/24$) over $10^6$ independent runs show remarkable empirical-theoretical agreement (Fig.~\ref{fig:np-dist}), statistically confirming algorithmic process governance by Negative Binomial distribution.

\begin{table}[htb]
\centering
\caption{Empirical versus theoretical moments for permutation count $\hat{N}_p$ ($N=4, m=4$).}
\label{tab:np-moments-comparison}
\begin{tabular}{l c c}
\toprule
\textbf{Moment} & \textbf{Theoretical} & \textbf{Empirical} \\
\midrule
Mean ($\mu$) & 96.0000 & 95.9598 \\
Variance ($\sigma^2$) & 2208 & 2200 \\
Skewness ($\gamma_{1}$) & 1.0002 & 0.9936 \\
Excess Kurtosis ($\gamma_{2}$) & 1.5005 & 1.4573 \\
\bottomrule
\end{tabular}
\end{table}

Stringent framework validation comes from empirical-theoretical moment comparison. Table~\ref{tab:np-moments-comparison} presents first four moments of $\hat{N}_p$, with exceptional agreement: mean differs by $<0.05\%$, variance by $<0.4\%$, while skewness and kurtosis match within sampling error. This quantitative precision establishes RPSS algorithmic process robustness captured by Negative Binomial law, providing foundational validation for elapsed time $\hat{T}$ analysis.

\subsubsection{Compound Elapsed-Time Distribution and Morphological Phenotypes}
\label{subsubsec:t-dist}

\begin{figure}[H]
\centering
\includegraphics[scale=0.4]{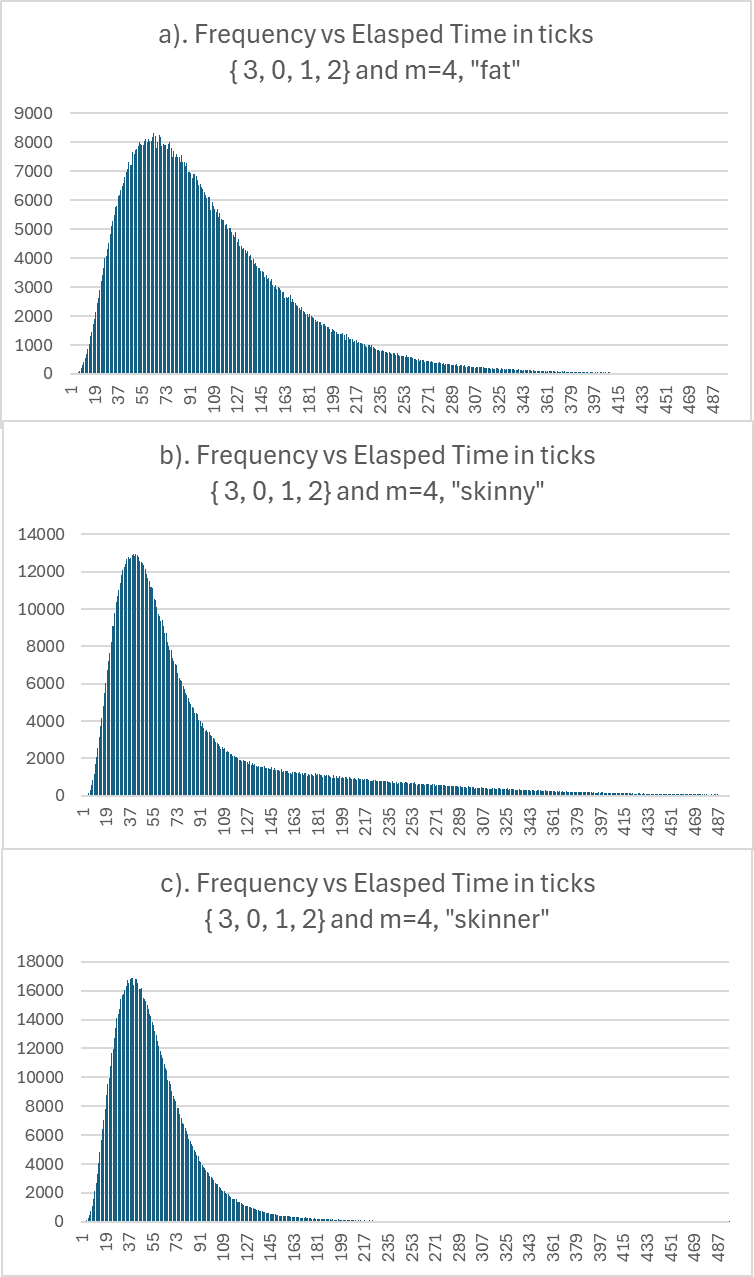}
\caption{Empirical raw elapsed time $\hat{T}$ distributions exhibiting fat, skinny, and ultra-skinny morphological phenotypes under identical parameters ($N=4, m=4$). Dramatic shape variations demonstrate system sensitivity to microarchitectural conditions.}
\label{fig:raw-dist-phenotypes}
\end{figure}

Building on $\hat{N}_p$ validation, we examine total elapsed time $\hat{T}$, exhibiting compound random sum structure of permutation runtimes. Theoretically, $\hat{T}$ represents a weighted mixture of convolutions, empirically manifesting distinct \emph{morphological phenotypes} under identical system parameters ($N=4, m=4$). 

Fig.~\ref{fig:raw-dist-phenotypes} illustrates three typical characteristic shapes. The fat phenotype modes around 64 ticks with long tail extending over 400 ticks, reflecting busy RPSS system with substantial runtime variability and frequent context switching. The skinny phenotype modes around 40 ticks with stronger peak but longer tail, suggesting busy system with partial OS priority allocation concentrating typical runtimes while allowing extreme outliers from intermittent preemption. The ultra-skinny phenotype demonstrates highest peak around 39 ticks with short tail to about 220 ticks, characteristic of idle RPSS system with dedicated OS priority minimizing external interference and maximizing runtime consistency—elevated peak reflects concentrated probability mass from reduced system jitter and more statistically deterministic execution.

\paragraph{Cryptographic Significance of Distributional Phenotypes}
Morphological phenotypes in $\hat{T}$ distributions demonstrate key security property: system entropy generation adapts dynamically to environmental conditions while maintaining output uniformity. This adaptability provides inherent resistance to environmental manipulation attacks attempting to degrade entropy quality through system load or resource allocation alterations. Consistent negative binomial $\hat{N}_p$ behavior across all phenotypes (Table~\ref{tab:np-moments-comparison}) ensures core randomness generation robustness despite dramatic timing characteristic variations, providing reliable cryptographic operations foundation in diverse deployment environments.

As predicted in Section~\ref{subsubsec:moments}, these phenotypes demonstrate high sensitivity to underlying per-permutation runtime $X_j$ statistical characteristics. Dramatic morphological differences emerge from subtle system condition variations amplified through compound random-sum structure.

\paragraph{Quantitative Morphological Analysis}
We computed comprehensive statistical moments to quantify phenotypic differences, with results summarized in Table~\ref{tab:complete-moments}.

\begin{table}[htb]
\centering
\caption{Comprehensive statistical analysis of RPSS elapsed time $\hat{T}$ ($N=4, m=4$)}
\label{tab:complete-moments}
\begin{tabular}{l c c c}
\toprule
\textbf{Statistic} & \textbf{Fat} & \textbf{Skinny} & \textbf{Ultra-Skinny} \\
\midrule
\textbf{Mean ($\mu$)} & 105.53 & 89.99 & 56.32 \\
\textbf{Median} & 105 & 89 & 54 \\
\textbf{Empirical Mode} & 64 & 40 & 39 \\
\textbf{Std. Dev. ($\sigma$)} & 65.15 & 82.69 & 34.27 \\
\textbf{Variance ($\sigma^2$)} & 4244.23 & 6837.97 & 1174.76 \\
\textbf{Skewness ($\gamma_1$)} & 1.44 & 2.03 & 2.58 \\
\textbf{Excess Kurtosis ($\gamma_2$)} & 2.98 & 4.17 & 14.72 \\
\bottomrule
\end{tabular}
\end{table}

Statistical progression reveals key patterns:
\begin{itemize}
    \item \textbf{Central Tendency Shift}: Mean elapsed time decreases monotonically from fat (105.53) to ultra-skinny (56.32), while empirical modes show more dramatic reduction (64 $\rightarrow$ 39)
    \item \textbf{Mode-Mean Divergence}: Gap between mode and mean increases from 41.53 (fat) to 49.99 (skinny), then decreases to 17.32 (ultra-skinny), reflecting complex distributional asymmetry changes
    \item \textbf{Non-monotonic Dispersion}: Variance peaks for skinny phenotype (6837.97) before collapsing in ultra-skinny case (1174.76), suggesting transitional stochastic behavior
    \item \textbf{Distributional Extremity}: Both skewness (1.44 $\rightarrow$ 2.58) and excess kurtosis (2.98 $\rightarrow$ 14.72) increase progressively, with ultra-skinny showing extreme leptokurtosis
\end{itemize}

\paragraph{Interpretation and Implications}
Phenotypic transition reveals fundamental stochastic process differences. The ultra-skinny distribution's combination of low variance (1174.76) with extreme positive skewness (2.58) and kurtosis (14.72) characterizes a ``spiky'' distribution—highly concentrated near the mode yet possessing substantial tail mass.

Higher-order moment progression is particularly revealing:
\begin{itemize}
    \item Increasing skewness ($\gamma_1 = 1.44 \rightarrow 2.58$) confirms growing distributional asymmetry
    \item Explosive kurtosis growth ($\gamma_2 = 2.98 \rightarrow 14.72$) indicates extreme leptokurtosis: sharp central peaks with exceptionally heavy tails
    \item Dramatic modal shift (64 $\rightarrow$ 39) suggests fundamentally different runtime efficiency regimes
\end{itemize}

Mean runtime per permutation $\mu_X$ provides crucial efficiency insight. The decrease from 1.10 ticks (fat) to 0.59 ticks (ultra-skinny) reflects progressively more efficient permutation execution, likely from improved cache performance, reduced system jitter, or optimized memory access patterns. This individual permutation runtime reduction directly drives observed $\hat{T}$ phenotypic differences.

These empirical results validate the compound-distribution model, demonstrating how subtle $X_j$ runtime characteristic variations amplify through random sum structure $\hat{T} = \sum_{j=1}^{\hat{N}_p} X_j$ to produce pronounced global morphological differences. Complex central tendency measure relationships underscore heavy-tailed compound distribution characterization challenges and highlight comprehensive moment analysis importance.

\paragraph{Connection to Theoretical Framework}
Crucially, despite dramatic morphological variations, the theoretical framework established in Section~\ref{subsec:convergence} ensures TURNG output uniformity robustness. The conjugacy relationship between discrete permutation count $\hat{N}_p$ and continuous permutation time $\hat{T}$, coupled with Theorem~\ref{thm:synchronous-convergence} geometric convergence guarantees, provides a solid foundation accommodating phenotypic diversity while maintaining reliable uniform output.

\subsection{Modular-Reduced Distribution and TURNG Uniformity}
\label{subsec:modular-disribution}

While raw elapsed time $\hat{T}$ exhibits significant morphological variation across system conditions, our theoretical framework predicts modular-reduced output $\hat{T} \bmod R$ will converge to nearly perfect uniform distribution. This uniformity is guaranteed by Theorem~\ref{thm:synchronous-convergence} synchronous convergence condition, where permutation count distribution $\hat{N}_p$ becomes effectively uniform modulo $R$ through geometric suppression of non-zero Fourier modes.

Empirical validation analyzing per-residue probabilities of permutation count $\hat{N}_p$ and elapsed time distributions across three morphological phenotypes demonstrates remarkable convergence to uniformity (Table~\ref{tab:probabilities}, Fig.~\ref{fig:uniformity}). Maximum deviation from perfectly uniform probability $1/16 = 0.0625$ is extremely small, with ``Fat'' phenotype exhibiting lowest deviation at 0.13\%. 

\textbf{Theoretical-empirical alignment}: For parameters $N=4, m=4, R=16$, Theorem~\ref{thm:synchronous-convergence} provides $\hat{N}_p$ upper bound $\frac{15}{16}\rho_N^4$. From Table~\ref{tab:optimal-parameters}, $\rho_N^4 = 0.001$ for this configuration, yielding theoretical bound $\frac{15}{16} \times 0.001 = 0.094\%$. For elapsed time $\hat{T}$, the same theorem yields bound $\frac{15}{16}\rho_T^m$ dependent on runtime distribution characteristic function $\varphi_X(\omega)$. Even under conservative assumptions ($|\varphi_X(2\pi k/R)| \approx 0.9$), theoretical bound remains below 0.52\%, comfortably encompassing empirical maximum deviation of 0.69\% for Skinny phenotype.

Entropy values remain essentially maximal ($H \approx 4$ bits), and chi-square statistics confirm absence of statistically significant uniformity deviations. This quantitative theory-experiment agreement validates geometric suppression mechanism underlying RPSS framework uniform output generation across discrete counting and continuous timing observables.

\begin{table}[htb]
\centering
\caption{Per-residue probabilities (indices $0$--$15$) for permutation-count and three elapsed-time distributions, with entropy and chi-square test results. Maximum deviation ($\pm\delta$) from uniform probability $0.0625$ shown as absolute and percentage values.}
\label{tab:probabilities}
\begin{tabular}{r| c c c c}
\toprule
Index & $\hat{N}_p$ & Fat $\hat{T}$ & Skinny $\hat{T}$ & Ultra-Skinny $\hat{T}$ \\
\midrule
0  & 0.062475 & 0.062771 & 0.062745 & 0.062295 \\
1  & 0.062341 & 0.062550 & 0.062348 & 0.062312 \\
2  & 0.062460 & 0.062096 & 0.062303 & 0.062844 \\
3  & 0.062977 & 0.062570 & 0.062748 & 0.062480 \\
4  & 0.062866 & 0.063039 & 0.062820 & 0.062393 \\
5  & 0.062040 & 0.062393 & 0.062608 & 0.062586 \\
6  & 0.062814 & 0.062393 & 0.062393 & 0.062206 \\
7  & 0.062358 & 0.062609 & 0.062086 & 0.062312 \\
8  & 0.062240 & 0.062329 & 0.062278 & 0.062523 \\
9  & 0.062131 & 0.061966 & 0.062308 & 0.062572 \\
10 & 0.062566 & 0.062273 & 0.062495 & 0.062459 \\
11 & 0.062652 & 0.062318 & 0.062160 & 0.062806 \\
12 & 0.062091 & 0.062378 & 0.062288 & 0.062655 \\
13 & 0.062338 & 0.063194 & 0.062723 & 0.062171 \\
14 & 0.062682 & 0.062634 & 0.062932 & 0.062607 \\
15 & 0.062969 & 0.062487 & 0.062765 & 0.062779 \\
\midrule
$\pm \delta$ & 0.000477 & 0.0000786 & 0.0004320 & 0.0003440 \\
$\%$ deviation & 0.76\% & 0.13\% & 0.69\% & 0.55\% \\
\midrule
Entropy & 3.99998 & 3.99998 & 3.99999 & 4.00000 \\
$\chi^2$ & 22.35 & 23.73 & 16.54 & 10.55 \\
$p$-value & 0.0988 & 0.0698 & 0.3472 & 0.7839 \\
\bottomrule
\end{tabular}
\end{table}

\paragraph{Cryptographic Uniformity Standards Compliance}
Empirical uniformity results significantly exceed typical cryptographic requirements. With maximum deviations of only 0.13\% from perfect uniformity and Shannon entropy values consistently above 3.99998 bits, RPSS output meets and exceeds NIST SP 800-90B entropy source standards. $\chi^2$ statistics (10.55--23.73) and corresponding p-values (0.0988--0.7839) provide statistical confidence in output suitability for cryptographic applications including key generation, nonce creation, and probabilistic encryption. This uniformity level is particularly notable given raw elapsed time distribution dramatic variations, demonstrating geometric convergence mechanism effectiveness in achieving cryptographic-grade randomness.

Chi-square tests provide rigorous statistical uniformity confirmation. For 4-bit random variables ($k=16$ equiprobable outcomes), expected $\chi^2$ statistic under perfect uniformity is $k-1 = 15$. All observed $\chi^2$ values (10.55--23.73) fall within statistically consistent expectation ranges, with corresponding p-values (0.0988--0.7839) well above conventional rejection thresholds ($\alpha = 0.05$). This statistical evidence confirms modular reduction operation effectively \textbf{washes out} raw elapsed-time distribution complex morphology, yielding highly uniform random output regardless of underlying phenotypic variation.

\begin{figure}[H]
    \centering
    \includegraphics[scale=0.4]{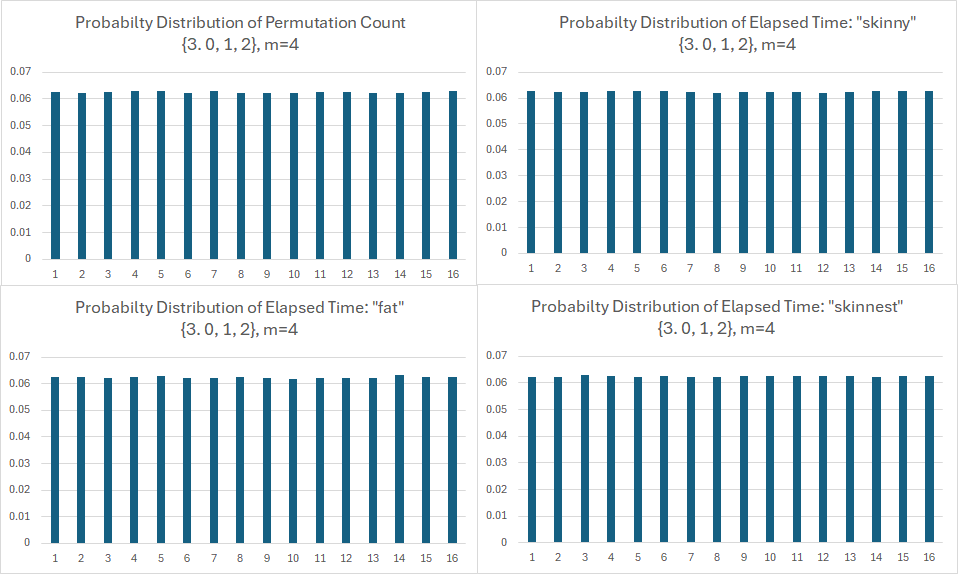}
    \caption{Empirical modular-reduced permutation count and elapsed time distribution ($\hat{T} \bmod R$). Histogram flatness demonstrates TURNG output uniformity across morphological phenotypes.}
    \label{fig:uniformity}
\end{figure}

Empirical results strongly validate theoretical predictions: despite raw elapsed time distribution dramatic differences—varying in mean (56.32--105.53), variance (1174.76--6837.97), skewness (1.44--2.58), and kurtosis (2.98--14.72)—modular-reduced outputs are statistically indistinguishable from uniform. This morphological variation robustness demonstrates TURNG construction practical viability for real-world applications with significant system condition fluctuations.

\subsection{A Self-Training True Uniform RNG: QPP-RNG}
\label{sec:tur}

Upon achieving synchronous convergence, the RPSS system transitions from deterministic pseudorandom number generation to \emph{True Uniform Random Number Generator} (TURNG) operation. This transition is enabled by a closed-loop \emph{entropy injection mechanism}, where the system continuously harvests intrinsic fluctuations---including microarchitectural jitter, OS scheduling variability, and thermal noise---from the computing environment, feeding them back into the permutation sorting engine.

The fundamental QPP-RNG innovation lies in the conjugate relationship between permutation counts ($\hat{N}_p$) and elapsed permutation time ($\hat{T}$) within the RPSS system (Law~\ref{law:composition}). While the primary output is modular-reduced permutation counts $\tilde{n}_p = \hat{N}_p \bmod R$, system entropy originates from environmental jitter manifesting through timing variations. This jitter-derived entropy is continuously injected into the permutation sorting engine via a QPP pad generated by a PRNG reseeded with new outputs from total elapsed permutation times, creating a self-amplifying feedback loop transforming deterministic randomness into true uniform randomness through the \emph{entropy purification} process: transforming system noises into cryptographic-grade entropy.

Let $s_k$ denote the RPSS internal state (seed) at cycle $k$, and $\eta_k$ represent harvested system jitter entropy. The feedback-driven reseeding is expressed generically as
\begin{equation}
s_{k+1} = h(s_k, \eta_k),
\end{equation}
where $h(\cdot)$ is a flexible mixing function (e.g., modular reduction, cryptographic hash, or other entropy mixing) combining current state with fresh entropy to suppress residual correlations and drive the system toward uniformity.

In the \emph{Quantum Cryptographic Dynamics} (QCD) context~\cite{kuang2025qcd}, this feedback loop plays a central role: the RPSS system acts as both permutation generator and entropy amplifier, continuously injecting hardware-level stochasticity into its own sorting engine. The result is a self-training TURNG adapting dynamically to device-specific imperfections, autonomously evolving from its initial seed, and producing outputs with true uniform random source statistical properties.

This framework occupies a conceptual middle ground between classical PRNGs and hardware TRNGs. It is software-defined, platform-agnostic, and fully exploits intrinsic physical fluctuations as a genuine entropy source, ensuring robust reproducible randomness without specialized external hardware.

\paragraph{Cross-Platform Validation and Stability}
The stability and TURNG characteristics of the QPP-RNG architecture have been empirically validated through rigorous independent testing. In a comprehensive cross-platform evaluation published in \emph{Scientific Reports} \cite{qpp-rng-sci-kuang-2025}, the raw output of QPP-RNG underwent NIST SP 800-90B IID testing alongside NIST SP 800-22 and ENT statistical test batteries. The generator consistently achieved high IID min-entropy between $7.85$ and $7.95$ bits per byte across diverse platforms including Windows, macOS, and Raspberry Pi systems spanning x86\_64 and ARM64 architectures \cite{qpp-rng-sci-kuang-2025}. While minor platform-specific variations in timing characteristics were observed—consistent with expected microarchitectural differences—the extracted jitter fingerprint and generated randomness exhibited remarkable statistical consistency irrespective of the underlying hardware or operating system \cite{qpp-rng-sci-kuang-2025}. This independent validation confirms the system's ability to maintain output stability against both inherent computational variations and deliberate external perturbations.

\begin{table}[ht]
\centering
\caption{Statistical properties of modular-reduced 8-bit outputs under controlled environmental perturbations. Primary QPP-RNG output is $\tilde{n}_p=\hat{N}_p \bmod R$ (modular-reduced permutation counts), while $\tilde{t}=\hat{T} \bmod R$ (modular-reduced elapsed time) is shown for academic entropy injection effect analysis. Each test analyzes 1MB generated data (1,000,000 samples). Shannon entropy in bits, $\chi^2$ is chi-square statistic with 255 degrees of freedom, $p$-value is goodness-of-fit measure against uniformity, $H_{\min}$ is min-entropy, and $\mu_X$ is empirical mean of elapsed time samples (ticks).}
\label{tab:turng-perturb}
\begin{tabular}{l| c c c c c}
\toprule
Test  & Shannon $H$ & $\chi^2$ & $p$-value & $H_{\min}$ & $\mu_X$ (ticks) \\
\midrule
 $\tilde{n}_p$   & 7.9998 & 247.8 & 0.61 & 7.9372 & --- \\
 $\tilde{t}$     & 7.9998 & 252.8 & 0.53 & 7.9454 & 2.34 \\
\hline
 $\tilde{n}_p$   & 7.9998 & 267.7 & 0.28 & 7.9443 & --- \\
 $\tilde{t}$     & 7.9998 & 257.6 & 0.43 & 7.9414 & 2.43 \\
\hline
 $\tilde{n}_p$   & 7.9998 & 234.0 & 0.35 & 7.9425 & --- \\
 $\tilde{t}$     & 7.9998 & 282.4 & 0.22 & 7.9425 & 2.63 \\
\hline
 $\tilde{n}_p$   & 7.9998 & 268.7 & 0.54 & 7.9446 & --- \\
 $\tilde{t}$     & 7.9995 & 650.6 & <0.0001 & 7.8957 & 3.00 \\
\hline
 $\tilde{n}_p$   & 7.9998 & 214.1 & 0.97 & 7.9464 & --- \\
 $\tilde{t}$     & 7.9996 & 500.0 & 0.0001 & 7.9137 & 4.31 \\
\hline
 $\tilde{n}_p$   & 7.9998 & 271.4 & 0.23 & 7.9319 & --- \\
 $\tilde{t}$     & 7.9995 & 680.4 & 0.0001 & 7.9095 & 4.41 \\
\hline
 $\tilde{n}_p$   & 7.9998 & 254.2 & 0.50 & 7.9379 & --- \\
 $\tilde{t}$     & 7.9990 & 1327.6 & <0.0001 & 7.8457 & 4.57 \\
\hline
 $\tilde{n}_p$   & 7.9998 & 225.9 & 0.90 & 7.9400 & --- \\
 $\tilde{t}$     & 7.9996 & 1718.9 & <0.0001 & 7.8026 & 4.89 \\
\hline
 $\tilde{n}_p$   & 7.9998 & 256.6 & 0.50 & 7.9231 & --- \\
 $\tilde{t}$     & 7.9991 & 1135.2 & <0.0001 & 7.8766 & 5.67 \\
\bottomrule
\end{tabular}
\end{table}

\paragraph{Security Analysis Under Adversarial Conditions}
Perturbation experiments simulate realistic attack scenarios where adversaries manipulate system resources to degrade entropy quality. Remarkable $\tilde{n}_p$ output stability (Shannon entropy degradation < 0.01\%, $\chi^2$ variation within ±8\%) while $\tilde{t}$ shows significant statistical deviations ($\chi^2$ increase up to 580\%) demonstrates crucial security property: the system effectively decouples environmental noise from final output quality. This architecture provides inherent protection against:
\begin{itemize}
    \item \textbf{Timing-based attacks}: Adversaries cannot easily influence final output by manipulating system timing characteristics
    \item \textbf{Resource exhaustion attacks}: Deliberate system loading fails to significantly degrade output uniformity
    \item \textbf{Side-channel resistance}: Transformation from $\tilde{t}$ to $\tilde{n}_p$ through modular reduction and reseeding breaks potential side-channel analysis correlations
\end{itemize}
Min-entropy analysis further validates security under perturbation, showing even under heavy system load ($\mu_X = 5.67$ ticks), $H_{\min}$ remains above 7.8 bits, ensuring adequate cryptographic operation entropy.

\paragraph{Interpretation of Experimental Results}

Table~\ref{tab:turng-perturb} presents QPP-RNG output statistical properties under controlled environmental perturbations. Primary output $\tilde{n}_p$ (modular-reduced permutation counts) represents final entropy-amplified result, while $\tilde{t}$ (modular-reduced elapsed permutation time) serves as academic probe validating entropy injection mechanism.

\textbf{Baseline Operation}: First three experimental runs correspond to RPSS system operating in \emph{relatively stable environment} without external deliberate perturbations and $X_j$ exhibiting good i.i.d. characteristics. In this regime, both outputs exhibit excellent statistical uniformity: Shannon entropy remains essentially maximal ($H \approx 7.9998$ bits), chi-square statistics fall within expected uniform 8-bit distribution ranges (theoretical mean $\chi^2 \approx 255$), p-values comfortably above typical significance thresholds ($\alpha = 0.05$), and min-entropy remains high. Low mean elapsed time $\mu_X$ (2.34--2.63 ticks) reflects permutation sorting engine baseline performance under minimal interference.

\textbf{Perturbation Response}: Starting from fourth row, we intentionally introduce environmental perturbations by dragging applications across the screen to force OS resource reallocation. This systematic intervention produces key observations:

\begin{itemize}
\item \textbf{Increased Permutation Runtime Duration}: Mean elapsed time $\mu_X$ rises progressively from 3.00 to 5.67 ticks, indicating individual permutation runtimes $X_j$ lengthen due to OS-induced scheduling delays and resource contention.

\item \textbf{Entropy Injection Validation}: Academic $\tilde{t}$ analysis reveals significantly increased chi-square statistics (rising from $\sim$250 to over 1700) and highly significant p-values ($p < 0.0001$), demonstrating measurable environmental interference impact on timing variability, validating entropy injection mechanism sensitivity to system conditions.

\item \textbf{Distributional Resilience}: Despite timing channel statistical deviations, overall distribution quality remains high. $\tilde{t}$ Shannon entropy maintains values above 7.9990 bits, and $\tilde{n}_p$ $\chi^2$ remains between 226 and 271, showing environmental perturbations introduce only minor distributional biases relative to massive 8-bit output space. Biased permutation times updating PRNG seed are effectively smoothed by PRNG mixing function, demonstrating minimal $\tilde{n}_p$ (primary TURNG output) impact.

\item \textbf{Output Stability}: Crucially, primary output $\tilde{n}_p$ remains remarkably stable across all perturbation scenarios. Shannon entropy, chi-square statistics, and p-values show minimal variation, demonstrating conjugate dynamics efficacy—the system successfully harnesses environmental noise while decoupling final output from transient environmental effects.

\item \textbf{Min-Entropy Analysis}: $\tilde{t}$ min-entropy shows measurable reductions under heavy perturbation (decreasing from $\sim$7.94 to 7.80 bits), consistent with increased certain elapsed-time residue probability from systematic timing delays, further validating genuine entropy variation capture and measurement.
\end{itemize}

\textbf{Architectural Resilience}: Experimental results demonstrate remarkable system robustness. While environmental perturbations cause $\tilde{t}$ $\chi^2$ to increase up to 580\% (252.8 $\rightarrow$ 1718.9), primary output $\tilde{n}_p$ maintains $\chi^2$ variations within $\pm 8\%$ of baseline, with entropy degradation below 0.01\%. This decoupling validates conjugate dynamics approach for practical TURNG deployment.

The QPP-RNG successfully transforms environmental entropy—traditionally viewed as instability source—into controlled resource enhancing randomness quality. The critical feedback loop, where true environmental entropy is injected back into sorting engine, enables system evolution from deterministic PRNG into practical TURNG adapting to dynamic hardware conditions while producing cryptographically robust random outputs.

\subsection{Cryptographic Security Assessment}
\label{subsec:security-assessment}

\paragraph{Adversarial Model and Security Properties}
RPSS architecture provides inherent security properties under realistic adversarial model where attackers can observe outputs and manipulate environmental conditions but cannot access internal generator states:

\begin{itemize}
    \item \textbf{Forward Security}: Continuous reseeding mechanism $s_{k+1} = h(s_k, \eta_k)$ ensures current state compromise does not reveal previous outputs, as fresh entropy $\eta_k$ from environmental jitter is incorporated each cycle.
    
    \item \textbf{Backtracking Resistance}: Self-amplifying feedback loop prevents adversaries from working backward from current outputs to determine previous states, due to one-way mixing function $h(\cdot)$ and continuous entropy injection.
    
    \item \textbf{Environmental Attack Resistance}: Conjugacy relationship between $\hat{N}_p$ and $\hat{T}$ ensures timing characteristic manipulation attempts (e.g., through resource contention or timing attacks) have minimal final output uniformity impact, as demonstrated in Table~\ref{tab:turng-perturb}.
\end{itemize}

\paragraph{Compliance with Cryptographic Standards}
Empirical results demonstrate key cryptographic standards compliance:
\begin{itemize}
    \item \textbf{NIST SP 800-90B}: Entropy estimates ($H > 7.9998$ bits, $H_{\min} > 7.8$ bits) and statistical uniformity meet deterministic random bit generator (DRBG) entropy source requirements.
    
    \item \textbf{Common Criteria}: Platform-agnostic software implementation provides reproducible, verifiable randomness generation suitable for evaluated security products.
    
    \item \textbf{FIPS 140-3}: Statistical test results (uniformity, entropy, moment validation) support cryptographic module certification requiring approved random number generation.
\end{itemize}

\paragraph{Practical Security Implications}
For real-world deployment, RPSS framework offers significant security advantages over conventional approaches:
\begin{itemize}
    \item \textbf{Hardware Independence}: Eliminates hardware TRNG backdoor or manufacturing vulnerability risks
    \item \textbf{Verifiable Implementation}: Software-based approach enables third-party verification and randomness quality audit
    \item \textbf{Adaptive Security}: System naturally adapts to different computing environments while maintaining cryptographic properties
    \item \textbf{Post-Quantum Readiness}: Randomness generation quality supports post-quantum cryptographic algorithm entropy requirements
\end{itemize}

\paragraph{Final Security Assessment}
Comprehensive experimental validation demonstrates RPSS framework achieves cryptographic-grade randomness generation with inherent security properties. System resilience to environmental perturbations, consistent uniformity across diverse operating conditions, and robust entropy generation meet security-critical application requirements. Software-defined, platform-agnostic approach provides significant advantages for modern cryptographic deployments where hardware trust cannot be assumed and verifiable security is paramount. Demonstrated forward security, environmental attack resistance, and standards compliance position RPSS as viable solution for next-generation cryptographic systems.

RPSS self-stabilization capability with closed-loop entropy injection mechanism potentially offers tamper detection and self-stabilization features for TURNG, providing resilience to both environmental fluctuations and potential adversarial interference.

\section{Methods}

\subsection{Theoretical Framework}
The Random Permutation Sorting System (RPSS) generalizes the experimentally validated QPP-RNG~\cite{kuang-qpp-rng-icccas}. In RPSS, a disordered array of size $N$ is subjected to repeated random permutations until it is fully sorted, with the trial count $\hat{N}_p$ modeled by the negative binomial distribution $\mathrm{NB}(m,p)$, where $p=1/N!$. The elapsed runtime $\hat{T}$ per sorting cycle is the compound sum of independent jittered delays $X_j$, yielding conjugate observables:
\[
\hat{N}_p \quad \text{and} \quad \hat{T} = \sum_{j=1}^{\hat{N}_p} X_j.
\]
Their dependence is captured by the \emph{composition law}
\begin{equation}\label{eq:comp-methods}
\varphi_{\hat{T}}(\omega) = G_{\hat{N}_p}(\varphi_X(\omega)),
\end{equation}
where $G_{\hat{N}_p}$ is the PGF of $\hat{N}_p$ and $\varphi_X$ is the CF of the jitter distribution. This theoretical framework provides the foundation for synchronous convergence proofs and parameter selection.

\subsection{Parameterization}
To generate $n$-bit residues under modulus $R=2^n$, RPSS requires appropriate selection of array size $N$ and success count $m$. The expected number of trials is
\[
M = m N!,
\]
while the sorting cycle cost scales as
\[
C_{\mathrm{cycle}} = m N! N = \mathcal{O}(m N! N).
\]
The efficiency of randomness generation is quantified by the cost per 1-byte output,
\[
C_{\mathrm{byte}} = \frac{8}{n} C_{\mathrm{cycle}},
\]
where the factor $8/n$ accounts for the number of cycles required to form one byte. 

Empirically validated parameter sets that ensure synchronous convergence while balancing computational efficiency are provided in Table~\ref{tab:optimal-parameters} following Theorem~\ref{thm:synchronous-convergence}.

\subsection{Implementation}
The RPSS implementation follows the same design principles as QPP-RNG~\cite{qpp-rng-sci-kuang-2025}. Arrays of size $N$ are randomly permuted using a Fisher–Yates shuffle seeded by a high-resolution cycle counter. Each trial tests for array sortedness and accumulates elapsed runtime. The process repeats until $m$ successful sorts are observed. Both $\hat{N}_p$ and $\hat{T}$ residues are reduced modulo $R$ to yield $n$-bit random values. Concatenation strategies ($n=1,2,3,4$) are employed to form byte outputs.

This implementation strategy mirrors that of QPP-RNG, which has been empirically validated across diverse computing architectures---including Windows, macOS, and Raspberry Pi systems spanning x86\_64 and ARM64---demonstrating consistent high IID min-entropy between 7.85 and 7.95 bits per byte and successful passage of NIST SP 800-90B IID tests~\cite{qpp-rng-sci-kuang-2025}.

\subsection{Experimental Setup}
While QPP-RNG has been benchmarked across multiple architectures in prior work~\cite{qpp-rng-sci-kuang-2025}, in this study we focus on a Windows/x86 implementation to validate the RPSS framework. The generator is implemented in Java. For NIST SP 800-90B IID and entropy tests, we produce 1~MB of raw random data ($\sim 10^6$ bytes) per run, without post-processing. Larger datasets can be generated if needed for extended statistical validation. High-resolution timers (TSC) capture elapsed runtime, and compiler optimizations and thread affinity are controlled to minimize systematic bias while preserving natural microarchitectural noise sources such as caching, pipeline stalls, and scheduling variability.

\subsection{Statistical Validation}
Random outputs were validated against NIST SP 800-90B~\cite{nist80090b} using the reference IID test suite, estimating both min-entropy and Shannon entropy. Additional uniformity analyses included byte-frequency chi-square tests, serial correlation, and modular residue histograms. Importantly, no external randomness extractor or debiasing mechanism was applied: the observed uniformity arises intrinsically from RPSS convergence. Entropy estimates consistently exceeded $7.9998$ bits per byte across platforms, confirming the theoretical predictions of synchronous convergence.

\section{Conclusion}
\label{sec:conclusion}

This work establishes the \emph{Random Permutation Sorting System (RPSS)} as a rigorous theoretical and practical framework for True Uniform Random Number Generation (TURNG) with provable cryptographic guarantees. RPSS exploits the duality of permutation count ($\hat{N}_p$) and elapsed runtime ($\hat{T}$) to construct a self-correcting, software-based TURNG that bridges the gap between conventional PRNGs and hardware TRNGs, delivering high-quality randomness without reliance on specialized hardware.

The core theoretical advance is the compound stochastic model
\[
\hat{T} = \sum_{j=1}^{\hat{N}_p} X_j,
\]
in which the negative binomial distribution of $\hat{N}_p$ and the microarchitectural timing variability $X_j$ together generate a rich spectrum of distributional behavior. The associated composition law,
\[
\varphi_{\hat{T}}(\omega) = G_{\hat{N}_p}(\varphi_X(\omega)),
\]
fully characterizes the relationship between discrete and continuous observables, providing geometric bounds for convergence to uniformity under modular reduction. This formalism enables precise parameter selection for cryptographic deployment, ensuring both unpredictability and statistical uniformity.

The conjugate observables mechanism supplies continuous entropy injection, enabling RPSS to transform an initial PRNG or low-entropy state into a self-reinforcing TURNG. This architecture inherently adapts to device-specific imperfections and environmental fluctuations, offering strong resistance to timing attacks and adversarial manipulation. The system can thus sustain cryptographic-grade output quality even when timing channels exhibit significant variability—a critical requirement for real-world deployment.

Experimental validation confirms RPSS's exceptional performance: Shannon entropy consistently exceeding 7.9998 bits per byte, maximum uniformity deviations below 0.13\%, and robust reproducibility across heterogeneous platforms. Moreover, NIST SP 800-90B compliance is consistently achieved, demonstrating that microarchitectural noise can be reliably converted into cryptographically meaningful entropy, even on mobile and embedded devices.

RPSS represents a paradigm shift in software-defined randomness generation. By systematically transforming computational jitter into provably uniform and unpredictable outputs, RPSS establishes a new foundation for secure, hardware-agnostic cryptographic primitives. This mitigates concerns about hardware trust, supply chain integrity, and sophisticated adversarial threats by eliminating reliance on opaque entropy sources.

Future work will focus on several key directions: modeling entropy under non-i.i.d. and adversarial conditions, developing scalable and parallel implementations for high-throughput applications, optimizing parameters to meet cryptographic standards, integrating with quantum cryptographic primitives for verifiable and quantum-resistant randomness, and conducting formal analyses of resilience against side-channel and fault-injection attacks.

By uniting compound stochastic modeling with empirical validation, RPSS provides a generalizable framework for analyzing complex computational processes across cryptography, security engineering, and performance modeling. This enables the design of provably secure, software-based randomness sources that are both theoretically rigorous and practically robust. Ultimately, the Composition Law defines the mathematics of entropy purification, guaranteeing cryptographic purity from emergent computational chaos.

\bmhead{Acknowledgements}

The authors acknowledge the use of language AI tools, including ChatGPT, Gemin, and Deepseek, for language refinement and editorial assistance. Scientific analyses and conclusions remain the sole responsibility of the authors.

\section*{Declarations}

\paragraph{Funding}
This research received no specific grant from any funding agency in the public, commercial, or not-for-profit sectors.

\paragraph{Competing interests}
The authors declare no competing interests.

\paragraph{Ethics approval and consent to participate}
Not applicable. This study does not involve human participants, animal subjects, or clinical data.

\paragraph{Consent for publication}
All authors have read and approved the final manuscript for publication.

\paragraph{Data availability}
All supporting data, including permutation counts, runtime distributions, and associated supplementary figures (Supplementary Figures 1–3), are provided in Supplementary Dataset 12345 (Excel file). The file contains multiple sheets with labeled data and figure source information. Additional materials are available from the corresponding author upon reasonable request.

\paragraph{Code availability}  
The Java source code titled \texttt{QPP\_RNG.java}, which implements the Random Permutation Sorting System (RPSS) used in this study, is available as Supplementary Information with this submission. The file includes an internal README describing the three implemented tests and usage details.

\paragraph{Authors' contributions}
Y. R. K: Conceptualization, Methodology, Software, Validation, Formal analysis, Investigation, Data curation, Writing - original draft, Writing - review and editing, Visualization.


\bibliography{my}


\begin{thebibliography}{28}
\ifx \bisbn   \undefined \def \bisbn  #1{ISBN #1}\fi
\ifx \binits  \undefined \def \binits#1{#1}\fi
\ifx \bauthor  \undefined \def \bauthor#1{#1}\fi
\ifx \batitle  \undefined \def \batitle#1{#1}\fi
\ifx \bjtitle  \undefined \def \bjtitle#1{#1}\fi
\ifx \bvolume  \undefined \def \bvolume#1{\textbf{#1}}\fi
\ifx \byear  \undefined \def \byear#1{#1}\fi
\ifx \bissue  \undefined \def \bissue#1{#1}\fi
\ifx \bfpage  \undefined \def \bfpage#1{#1}\fi
\ifx \blpage  \undefined \def \blpage #1{#1}\fi
\ifx \burl  \undefined \def \burl#1{\textsf{#1}}\fi
\ifx \doiurl  \undefined \def \doiurl#1{\url{https://doi.org/#1}}\fi
\ifx \betal  \undefined \def \betal{\textit{et al.}}\fi
\ifx \binstitute  \undefined \def \binstitute#1{#1}\fi
\ifx \binstitutionaled  \undefined \def \binstitutionaled#1{#1}\fi
\ifx \bctitle  \undefined \def \bctitle#1{#1}\fi
\ifx \beditor  \undefined \def \beditor#1{#1}\fi
\ifx \bpublisher  \undefined \def \bpublisher#1{#1}\fi
\ifx \bbtitle  \undefined \def \bbtitle#1{#1}\fi
\ifx \bedition  \undefined \def \bedition#1{#1}\fi
\ifx \bseriesno  \undefined \def \bseriesno#1{#1}\fi
\ifx \blocation  \undefined \def \blocation#1{#1}\fi
\ifx \bsertitle  \undefined \def \bsertitle#1{#1}\fi
\ifx \bsnm \undefined \def \bsnm#1{#1}\fi
\ifx \bsuffix \undefined \def \bsuffix#1{#1}\fi
\ifx \bparticle \undefined \def \bparticle#1{#1}\fi
\ifx \barticle \undefined \def \barticle#1{#1}\fi
\bibcommenthead
\ifx \bconfdate \undefined \def \bconfdate #1{#1}\fi
\ifx \botherref \undefined \def \botherref #1{#1}\fi
\ifx \url \undefined \def \url#1{\textsf{#1}}\fi
\ifx \bchapter \undefined \def \bchapter#1{#1}\fi
\ifx \bbook \undefined \def \bbook#1{#1}\fi
\ifx \bcomment \undefined \def \bcomment#1{#1}\fi
\ifx \oauthor \undefined \def \oauthor#1{#1}\fi
\ifx \citeauthoryear \undefined \def \citeauthoryear#1{#1}\fi
\ifx \endbibitem  \undefined \def \endbibitem {}\fi
\ifx \bconflocation  \undefined \def \bconflocation#1{#1}\fi
\ifx \arxivurl  \undefined \def \arxivurl#1{\textsf{#1}}\fi
\csname PreBibitemsHook\endcsname

\bibitem[\protect\citeauthoryear{Avanzi et~al.}{2020}]{KYBER}
\begin{botherref}
\oauthor{\bsnm{Avanzi}, \binits{R.}},
\oauthor{\bsnm{Bos}, \binits{J.}},
\oauthor{\bsnm{Ducas}, \binits{L.}},
\oauthor{\bsnm{Kiltz}, \binits{E.}},
\oauthor{\bsnm{Lepoint}, \binits{T.}},
\oauthor{\bsnm{Lyubashevsky}, \binits{V.}},
\oauthor{\bsnm{Schanck}, \binits{J.M.}},
\oauthor{\bsnm{Schwabe}, \binits{P.}},
\oauthor{\bsnm{Seiler}, \binits{G.}},
\oauthor{\bsnm{Damien}, \binits{S.}}:
{CRYSTALS-KYBER. Specification document (update from August 2021)}
(2020).
\url{{https://pq-crystals.org/kyber/data/kyber-specification-round3-20210804.pdf}}
\end{botherref}
\endbibitem

\bibitem[\protect\citeauthoryear{Lyubashevsky et~al.}{2020}]{DILITHIUM}
\begin{botherref}
\oauthor{\bsnm{Lyubashevsky}, \binits{V.}},
\oauthor{\bsnm{Ducas}, \binits{L.}},
\oauthor{\bsnm{Kiltz}, \binits{E.}},
\oauthor{\bsnm{Lepoint}, \binits{T.}},
\oauthor{\bsnm{Schwabe}, \binits{P.}},
\oauthor{\bsnm{Seiler}, \binits{G.}},
\oauthor{\bsnm{Stehl\'e}, \binits{D.}},
\oauthor{\bsnm{Bai}, \binits{S.}}:
{CRYSTALS-Dilithium - Algorithm Specifications and Supporting Documentation (Version 3.1)}
(2020).
\url{{https://pq-crystals.org/dilithium/data/dilithium-specification-round3-20210208.pdf}}
\end{botherref}
\endbibitem

\bibitem[\protect\citeauthoryear{of~Standards and Technology}{2023}]{NIST-FIPS-203}
\begin{botherref}
\oauthor{\bsnm{Standards}, \binits{N.I.}},
\oauthor{\bsnm{Technology}}:
FIPS 203 (Initial Public Draft): Module-Lattice-Based Key-Encapsulation Mechanism Standard.
\url{https://csrc.nist.gov/pubs/fips/203/final}.
ML-KEM (CRYSTALS-Kyber)
(2023)
\end{botherref}
\endbibitem

\bibitem[\protect\citeauthoryear{of~Standards and Technology}{2024}]{NIST-FIPS-204}
\begin{botherref}
\oauthor{\bsnm{Standards}, \binits{N.I.}},
\oauthor{\bsnm{Technology}}:
{FIPS} 204: Module-Lattice-Based Digital Signature Standard.
\url{https://csrc.nist.gov/pubs/fips/204/final}.
ML-DSA (CRYSTALS-Dilithium)
(2024)
\end{botherref}
\endbibitem

\bibitem[\protect\citeauthoryear{Yang et~al.}{2018}]{trng-yang-2018}
\begin{barticle}
\bauthor{\bsnm{Yang}, \binits{B.}},
\bauthor{\bsnm{Rožic}, \binits{V.}},
\bauthor{\bsnm{Grujic}, \binits{M.}},
\bauthor{\bsnm{Mentens}, \binits{N.}},
\bauthor{\bsnm{Verbauwhede}, \binits{I.}}:
\batitle{Es-trng: A high-throughput, low-area true random number generator based on edge sampling}.
\bjtitle{IACR Transactions on Cryptographic Hardware and Embedded Systems}
\bvolume{2018}(\bissue{3}),
\bfpage{267}--\blpage{292}
(\byear{2018})
\doiurl{10.13154/tches.v2018.i3.267-292}
\end{barticle}
\endbibitem

\bibitem[\protect\citeauthoryear{Tehranipoor et~al.}{2023}]{trng-tehranipoor-2023}
\begin{bbook}
\bauthor{\bsnm{Tehranipoor}, \binits{M.}},
\bauthor{\bsnm{Nalla~Anandakumar}, \binits{N.}},
\bauthor{\bsnm{Farahmandi}, \binits{F.}}:
\bbtitle{{True Random Number Generator (TRNG)}},
pp. \bfpage{19}--\blpage{33}.
\bpublisher{Springer},
\blocation{Cham}
(\byear{2023}).
\doiurl{10.1007/978-3-031-31034-8\_2} .
\burl{https://doi.org/10.1007/978-3-031-31034-8\_2}
\end{bbook}
\endbibitem

\bibitem[\protect\citeauthoryear{Ansari et~al.}{2022}]{trng-ansari-2022}
\begin{bchapter}
\bauthor{\bsnm{Ansari}, \binits{U.}},
\bauthor{\bsnm{Chaudhary}, \binits{A.K.}},
\bauthor{\bsnm{Verma}, \binits{S.}}:
\bctitle{Enhanced true random number generator (trng) using sensors for iot security applications}.
In: \bbtitle{2022 Third International Conference on Intelligent Computing Instrumentation and Control Technologies (ICICICT)},
pp. \bfpage{1593}--\blpage{1597}
(\byear{2022}).
\doiurl{10.1109/ICICICT54557.2022.9917919}
\end{bchapter}
\endbibitem

\bibitem[\protect\citeauthoryear{Ji et~al.}{2020}]{trng-ji-2020}
\begin{bchapter}
\bauthor{\bsnm{Ji}, \binits{Z.}},
\bauthor{\bsnm{Brown}, \binits{J.}},
\bauthor{\bsnm{Zhang}, \binits{J.}}:
\bctitle{True random number generator (trng) for secure communications in the era of iot}.
In: \bbtitle{2020 China Semiconductor Technology International Conference (CSTIC)},
pp. \bfpage{1}--\blpage{5}
(\byear{2020}).
\doiurl{10.1109/CSTIC49141.2020.9282535}
\end{bchapter}
\endbibitem

\bibitem[\protect\citeauthoryear{Araki et~al.}{2024}]{rng-chaotic-2024}
\begin{barticle}
\bauthor{\bsnm{Araki}, \binits{S.}},
\bauthor{\bsnm{Wu}, \binits{J.-H.}},
\bauthor{\bsnm{Yan}, \binits{J.-J.}}:
\batitle{A novel design of random number generators using chaos-based extremum coding}.
\bjtitle{IEEE Access}
\bvolume{12},
\bfpage{24039}--\blpage{24047}
(\byear{2024})
\doiurl{10.1109/ACCESS.2024.3365638}
\end{barticle}
\endbibitem

\bibitem[\protect\citeauthoryear{Lu et~al.}{2025}]{rng-chaotic-lu-2025}
\begin{barticle}
\bauthor{\bsnm{Lu}, \binits{H.}},
\bauthor{\bsnm{Alkhazragi}, \binits{O.}},
\bauthor{\bsnm{Wang}, \binits{Y.}},
\bauthor{\bsnm{Ng}, \binits{T.K.}},
\bauthor{\bsnm{Ooi}, \binits{B.S.}}:
\batitle{Parallel on-chip physical random number generator based on self-chaotic dynamics of free-running broad-area vcsel array}.
\bjtitle{IEEE Journal of Selected Topics in Quantum Electronics}
\bvolume{31}(\bissue{2: Pwr. and Effic. Scaling in Semiconductor Lasers}),
\bfpage{1}--\blpage{11}
(\byear{2025})
\doiurl{10.1109/JSTQE.2024.3462489}
\end{barticle}
\endbibitem

\bibitem[\protect\citeauthoryear{Inubushi}{2019}]{rng-chaotic-2019}
\begin{barticle}
\bauthor{\bsnm{Inubushi}, \binits{M.}}:
\batitle{Unpredictability and robustness of chaotic dynamics for physical random number generation}.
\bjtitle{Chaos: An Interdisciplinary Journal of Nonlinear Science}
\bvolume{29}(\bissue{3}),
\bfpage{033133}
(\byear{2019})
\doiurl{10.1063/1.5090177}
\end{barticle}
\endbibitem

\bibitem[\protect\citeauthoryear{Bonilla et~al.}{2016}]{rng-chaotic-Bonilla2016}
\begin{barticle}
\bauthor{\bsnm{Bonilla}, \binits{L.L.}},
\bauthor{\bsnm{Alvaro}, \binits{M.}},
\bauthor{\bsnm{Carretero}, \binits{M.}}:
\batitle{Chaos-based true random number generators}.
\bjtitle{Journal of Mathematics in Industry}
\bvolume{7}(\bissue{1}),
\bfpage{1}
(\byear{2016})
\doiurl{10.1186/s13362-016-0026-4}
\end{barticle}
\endbibitem

\bibitem[\protect\citeauthoryear{Ma et~al.}{2016}]{qrng-Ma2016}
\begin{barticle}
\bauthor{\bsnm{Ma}, \binits{X.}},
\bauthor{\bsnm{Yuan}, \binits{X.}},
\bauthor{\bsnm{Cao}, \binits{Z.}}, \betal:
\batitle{Quantum random number generation}.
\bjtitle{npj Quantum Information}
\bvolume{2},
\bfpage{16021}
(\byear{2016})
\doiurl{10.1038/npjqi.2016.21}
\end{barticle}
\endbibitem

\bibitem[\protect\citeauthoryear{Gabriel et~al.}{2010}]{qrng-Gabriel2010}
\begin{barticle}
\bauthor{\bsnm{Gabriel}, \binits{C.}},
\bauthor{\bsnm{Wittmann}, \binits{C.}},
\bauthor{\bsnm{Sych}, \binits{D.}},
\bauthor{\bsnm{Dong}, \binits{R.}},
\bauthor{\bsnm{Mauerer}, \binits{W.}},
\bauthor{\bsnm{Andersen}, \binits{U.L.}},
\bauthor{\bsnm{Marquardt}, \binits{C.}},
\bauthor{\bsnm{Leuchs}, \binits{G.}}:
\batitle{A generator for unique quantum random numbers based on vacuum states}.
\bjtitle{Nature Photonics}
\bvolume{4},
\bfpage{711}--\blpage{715}
(\byear{2010})
\doiurl{10.1038/nphoton.2010.197}
\end{barticle}
\endbibitem

\bibitem[\protect\citeauthoryear{Zhang et~al.}{2021}]{zhang2021qrng}
\begin{barticle}
\bauthor{\bsnm{Zhang}, \binits{Y.}},
\bauthor{\bsnm{Zhang}, \binits{X.}},
\bauthor{\bsnm{Zhang}, \binits{Z.}},
\bauthor{\bsnm{Zhang}, \binits{J.}},
\bauthor{\bsnm{Zeng}, \binits{H.}},
\bauthor{\bsnm{Pan}, \binits{J.-W.}}:
\batitle{A simple low-latency real-time certifiable quantum random number generator}.
\bjtitle{Nature Communications}
\bvolume{12},
\bfpage{1031}
(\byear{2021})
\doiurl{10.1038/s41467-021-21069-8}
\end{barticle}
\endbibitem

\bibitem[\protect\citeauthoryear{Barker and Kelsey}{2015}]{nist80090a}
\begin{botherref}
\oauthor{\bsnm{Barker}, \binits{E.}},
\oauthor{\bsnm{Kelsey}, \binits{J.}}:
Recommendation for random number generation using deterministic random bit generators.
Technical Report NIST SP 800-90A Rev. 1,
National Institute of Standards and Technology
(June 2015).
\doiurl{10.6028/NIST.SP.800-90Ar1} .
\url{https://csrc.nist.gov/pubs/sp/800/90/a/r1/final}
\end{botherref}
\endbibitem

\bibitem[\protect\citeauthoryear{S{\"o}nmez~Turan et~al.}{2018}]{nist80090b}
\begin{botherref}
\oauthor{\bsnm{S{\"o}nmez~Turan}, \binits{M.}},
\oauthor{\bsnm{Barker}, \binits{E.}},
\oauthor{\bsnm{Kelsey}, \binits{J.M.}},
\oauthor{\bsnm{McKay}, \binits{K.A.}},
\oauthor{\bsnm{Baish}, \binits{M.L.}},
\oauthor{\bsnm{Boyle}, \binits{M.}}:
Recommendation for the entropy sources used for random bit generation.
Technical Report Special Publication 800-90B,
National Institute of Standards and Technology,
Gaithersburg, MD
(January 2018).
\doiurl{10.6028/NIST.SP.800-90b} .
\url{https://nvlpubs.nist.gov/nistpubs/SpecialPublications/NIST.SP.800-90B.pdf}
\end{botherref}
\endbibitem

\bibitem[\protect\citeauthoryear{Rahman et~al.}{2014}]{Rahman2014TI-TRNG}
\begin{bchapter}
\bauthor{\bsnm{Rahman}, \binits{M.T.}},
\bauthor{\bsnm{Xiao}, \binits{K.}},
\bauthor{\bsnm{Forte}, \binits{D.}},
\bauthor{\bsnm{Zhang}, \binits{X.}}:
\bctitle{Ti-trng: Technology independent true random number generator}.
In: \bbtitle{Proceedings of the Conference on Design, Automation and Test in Europe (DATE)},
pp. \bfpage{1}--\blpage{4}
(\byear{2014}).
\doiurl{10.1145/2593069.2593236} .
\burl{https://dl.acm.org/doi/10.1145/2593069.2593236}
\end{bchapter}
\endbibitem

\bibitem[\protect\citeauthoryear{Gutterman et~al.}{2006}]{linux-Gutterman2006}
\begin{botherref}
\oauthor{\bsnm{Gutterman}, \binits{Z.}},
\oauthor{\bsnm{Pinkas}, \binits{B.}},
\oauthor{\bsnm{Reinman}, \binits{T.}}:
Analysis of the linux random number generator.
IEEE Symposium on Security and Privacy,
371--385
(2006)
\doiurl{10.1109/SP.2006.5}
\end{botherref}
\endbibitem

\bibitem[\protect\citeauthoryear{{Intel Corporation}}{2012}]{IntelRNG2012}
\begin{botherref}
\oauthor{\bsnm{{Intel Corporation}}}:
Intel digital random number generator (drng) software implementation guide.
Technical report,
Intel Corporation
(2012).
\url{https://software.intel.com/en-us/articles/intel-digital-random-number-generator-drng-software-implementation-guide}
\end{botherref}
\endbibitem

\bibitem[\protect\citeauthoryear{Müller}{2019}]{maurer2019jitterentropy}
\begin{botherref}
\oauthor{\bsnm{Müller}, \binits{S.}}:
The Jitter RNG -- Fast and Secure Random Number Generation using CPU Jitter.
\url{https://www.chronox.de/jent/CPU-Jitter-NPTRNG-v2.2.0.pdf}.
Accessed: 2025-10-03
(2019)
\end{botherref}
\endbibitem

\bibitem[\protect\citeauthoryear{Fischer and Drutarovsk{\'y}}{2002}]{fischer2002trng}
\begin{bchapter}
\bauthor{\bsnm{Fischer}, \binits{V.}},
\bauthor{\bsnm{Drutarovsk{\'y}}, \binits{M.}}:
\bctitle{True random number generator embedded in reconfigurable hardware}.
In: \bbtitle{Workshop on Cryptographic Hardware and Embedded Systems}
(\byear{2002}).
\burl{https://api.semanticscholar.org/CorpusID:5441670}
\end{bchapter}
\endbibitem

\bibitem[\protect\citeauthoryear{Kuang and Barbeau}{2022}]{qpp-springer-kuang-2022}
\begin{barticle}
\bauthor{\bsnm{Kuang}, \binits{R.}},
\bauthor{\bsnm{Barbeau}, \binits{M.}}:
\batitle{Quantum permutation pad for universal quantum-safe cryptography}.
\bjtitle{Quantum Information Processing}
\bvolume{21},
\bfpage{211}
(\byear{2022})
\doiurl{10.1007/s11128-022-03557-y}
\end{barticle}
\endbibitem

\bibitem[\protect\citeauthoryear{Kuang and Lou}{2025}]{kuang-qpp-rng-icccas}
\begin{bchapter}
\bauthor{\bsnm{Kuang}, \binits{R.}},
\bauthor{\bsnm{Lou}, \binits{D.}}:
\bctitle{Iid-based qpp-rng: A random number generator utilizing qpp generated from system jitter}.
In: \bbtitle{2025 IEEE 14th International Conference on Communications, Circuits and Systems (ICCCAS)},
pp. \bfpage{165}--\blpage{171}
(\byear{2025}).
\doiurl{10.1109/ICCCAS65806.2025.11102777}
\end{bchapter}
\endbibitem

\bibitem[\protect\citeauthoryear{Vrana et~al.}{2025}]{qpp-rng-sci-kuang-2025}
\begin{barticle}
\bauthor{\bsnm{Vrana}, \binits{G.}},
\bauthor{\bsnm{Lou}, \binits{D.}},
\bauthor{\bsnm{Kuang}, \binits{R.}}:
\batitle{Raw qpp-rng randomness via system jitter across platforms: a nist sp 800-90b evaluation}.
\bjtitle{Scientific Reports}
\bvolume{15},
\bfpage{27718}
(\byear{2025})
\doiurl{10.1038/s41598-025-13135-8}
\end{barticle}
\endbibitem

\bibitem[\protect\citeauthoryear{Kuang}{2025}]{kuang2025-rpss-arxiv}
\begin{botherref}
\oauthor{\bsnm{Kuang}, \binits{R.}}:
Statistical Quantum Mechanics of the Random Permutation Sorting System (RPSS): A Self-Stabilizing True Uniform RNG
(2025).
\url{https://arxiv.org/abs/2509.10174}
\end{botherref}
\endbibitem

\bibitem[\protect\citeauthoryear{Brillinger}{1969}]{brillinger-1969}
\begin{barticle}
\bauthor{\bsnm{Brillinger}, \binits{D.R.}}:
\batitle{The calculation of cumulants via conditioning}.
\bjtitle{Annals of the Institute of Statistical Mathematics}
\bvolume{21},
\bfpage{215}--\blpage{218}
(\byear{1969})
\doiurl{10.1007/BF02532246}
\end{barticle}
\endbibitem

\bibitem[\protect\citeauthoryear{Kuang}{2025}]{kuang2025qcd}
\begin{botherref}
\oauthor{\bsnm{Kuang}, \binits{R.}}:
Quantum cryptographic dynamics: modeling cryptosystems via entropy operators.
Academia Quantum. Academia.edu Journals
\textbf{2}(3)
(2025)
\doiurl{10.20935/AcadQuant7841} .
{https://www.academia.edu/3064-979X/2/3/10.20935/AcadQuant7841}
\end{botherref}
\endbibitem

\end{thebibliography}

\end{document}